\documentclass[11pt]{article}
\usepackage{epsfig,amsmath,amsfonts,amssymb,url,fullpage}
\usepackage{verbatim}
\usepackage{enumerate}
\usepackage{color}
\usepackage[bookmarks=false]{hyperref}

\newcommand{\etal}{{\it et~al.\ }}

\newcommand{\mycase}[1]{{\underline{Case~#1}:}}

\newtheorem{theorem}{Theorem}

\newtheorem{lemma}[theorem]{Lemma}

\newtheorem{fact}[theorem]{Fact}

\newenvironment{proof}{{\noindent\bf Proof:\/}}{$\Box$\vskip 0.1in}
\newcommand{\qed}{}

\newcommand{\calA}{{\cal A}}

\newcommand{\calS}{{\cal S}}
\newcommand{\calR}{{\cal R}}

\newcommand{\e}{\mathrm{e}}

\newcommand{\braced}[1]{{\left\{#1\right\}}}

\newcommand{\brackd}[1]{{\left[#1\right]}}

\newcommand{\suchthat}{{\,:\,}}
\newcommand{\before}{\lhd}
\newcommand{\beforeeq}{\unlhd}
\newcommand{\after}{\rhd}
\newcommand{\aftereq}{\unrhd}

\newcommand{\dele}{{\mbox{\tt d}}}
\newcommand{\take}{{\mbox{\tt t}}}

\newcommand{\Greedy}{{\mbox{\sf Greedy}}}
\newcommand{\Rmix}{{\mbox{\sf RMix}}}
\newcommand{\UniRand}{{\mbox{\sf UniRand}}}

\newcommand{\DecQueEFH}{\mbox{\sf DecQueEFH}}

\newcommand{\MarkAndPick}{{\mbox{\sf MarkAndPick}}}

\newcommand{\FIFOQueEH}{{\mbox{\sf FIFOQueEH}}}

\newcommand{\myif}{{\underline{if}}}
\newcommand{\mythen}{{\underline{then}}}
\newcommand{\myelse}{{\underline{else}}}

\definecolor{myBlue}{rgb}{0.5,0.5,1}
\definecolor{myRed}{rgb}{0.9,0.3,0.1}
\definecolor{myYellow}{rgb}{0.9,0.9,0}
\definecolor{myGreen}{rgb}{0.1,1,0}

\title{Generalized Whac-a-Mole}

\author{%
Marcin Bienkowski\thanks{Institute of Computer Science,
University of Wroc{\l}aw,
50-383 Wroc{\l}aw, Poland.
Supported by MNiSW grants number N206 001 31/0436, 2006--2008 and N N206 1723 33, 2007--2010.}
\and
Marek Chrobak\footnote{Department of Computer Science,
University of California, Riverside, CA 92521, USA.
Supported by NSF grants OISE-0340752 and CCF-0729071.
}
\and
Christoph D\"urr\footnote{CNRS, LIX UMR 7161, Ecole Polytechnique
91128 Palaiseau, France.
Supported by ANR Alpage.}
\and
Mathilde Hurand\footnotemark[3]
\and
Artur Je{\.z}\footnotemark[1]
\and
{\L}ukasz Je{\.z}\footnotemark[1]
\and
Jakub {\L}opusza{\'n}ski\footnotemark[1]
\and
Grzegorz Stachowiak\footnotemark[1]
}

\date{}

\begin{document}

\maketitle

\begin{abstract}
We consider online competitive algorithms
for the problem of collecting weighted items from a
dynamic set $\calS$, when items are added to or deleted from $\calS$
over time. The objective is to maximize the total weight of
collected items. We study the general version, as well
as variants with various restrictions, including the
following: the \emph{uniform case}, when all items have
the same weight, the \emph{decremental sets}, when
all items are present at the beginning and only deletion
operations are allowed, and \emph{dynamic queues}, where
the dynamic set is ordered and only its prefixes can
be deleted (with no restriction on insertions).
The dynamic queue case is a generalization of bounded-delay
packet scheduling (also referred to as buffer management).
We present several upper and lower bounds on the competitive
ratio for these variants.
\end{abstract}


\section{Introduction}



Whac-a-mole is an old arcade game, where plastic ``moles'' pop out
of holes in the machine for short periods of time, in some
unpredictable way, and the player uses a mallet to ``whack'' as
many moles as possible.\footnote{No mole was harmed in the course of this research.}
In the generalized version that we
consider, multiple moles may be present at the same time, and
different moles may have different values.

In a more formal setting, we think of it as a dynamic set $\calS$ of
weighted items (moles).
Before each step, some items
can be deleted from $\calS$ and other items can be added to $\calS$.
We are allowed to collect one item (whack one mole) from $\calS$ per step.
Each item can be collected only once (the mallet does its job). The objective is to maximize the total weight of the collected items.

To our knowledge, this simple and fundamental problem has
not been explicitly addressed in the literature.
By placing appropriate assumptions on the
structure of $\calS$ or on the type of allowed operations,
one can obtain a number of natural special cases,
some of which are related to known online problems.

We study the general case of dynamic sets, as well
as some restricted cases, among which we focus on the
following versions:

\emph{Dynamic Queue:} In this case, $\calS$ represents a list, i.e.
the items in $\calS$ are ordered. Items can be added to $\calS$ at
any place, but only a prefix of $\calS$ can be deleted.
A queue is called a \emph{FIFO queue} if insertions are allowed
only at the end.

\emph{Decremental Sets:} Here, all items are
added at the beginning, and only deletions are
allowed afterwards. In particular, one can consider
the case of decremental queues, where only prefix-deletion
operations are allowed.

The case of dynamic queues generalizes the well-studied
problem of bounded-delay packet scheduling (a.k.a.
buffer management), or, equivalently, the
problem of scheduling unit jobs with deadlines for
maximum weighted throughput. In this problem, packets
with values and deadlines arrive in a buffer of a network link. At each step,
we can send one packet along the link. The objective is to maximize the
total value of packets sent before their deadlines. This
is a special case of our dynamic queue problem, where packets are
represented by items ordered according to deadlines.
The difference is crucial though: in packet scheduling,
packet arrival times are unknown but their departure
times are known, while in dynamic queues \emph{both} the arrival
and departure times are not known.
The FIFO case generalizes the version of packet
scheduling with agreeable deadlines.


Competitive algorithms for various versions of bounded-delay packet
scheduling problem have been extensively studied
\cite{AnMaZh03,CCFJST06,ChiFun03b,Hajek01,KLMPSS04,KeMaSt05,LiSeSt07}.
In particular, it is known that no deterministic online algorithm can have
competitive ratio better than $\phi\approx 1.618$
\cite{AnMaZh03,ChiFun03b}, and algorithms with competitive ratio
$\approx 1.83$ have been recently developed \cite{LiSeSt07,EngWes07}.
These are the best lower and upper bounds
for this problem and closing this gap remains a tantalizing open problem.
For agreeable deadlines, an upper bound of $\phi$ has been
established in \cite{LiSeSt05}.

The dynamic set problem has some indirect connections to other
packet scheduling problems, where the objective is
to maximize the value of packets reaching their destination under
various scenarios
\cite{AzaRic04B,BFKMSS04,KLMPSS04,KeMaSt05}.

A different, metric version of whac-a-mole was considered in
\cite{GuKrMeVr06}.
In this approach moving the mallet to a mole takes time and the duration of moles'
exposure is known. However, their results are inapplicable in our model.


\smallskip
\emph{Our results.}
We first consider the general case of dynamic sets.
For deterministic algorithms, it is quite easy to show
a lower bound of $2$ (even for two items in the decremental
case), and this ratio can
be achieved by a simple greedy algorithm that always
collects the heaviest item.

For randomized algorithms, we focus on the
\emph{uniform case}, when all items have weight $1$.
We show that no memoryless
randomized algorithm can achieve competitive ratio better than $2$.
We then study the uniform decremental case, for which we
give an online randomized algorithm with competitive
ratio $\e/(\e-1)$ (against an oblivious adversary) and prove
a matching lower bound.

Most of our results concern dynamic queues. In the deterministic case,
even for decremental queues, it is
quite easy to show a lower bound of $\phi$. We improve
this bound, by proving a lower bound of $\approx 1.63$, that
applies even to the decremental case.
We also show that no memoryless algorithm can have a ratio
smaller than $2$. This contrasts with the
result of Englert and Westerman~\cite{EngWes07} who
gave a memoryless algorithm for packet scheduling with
ratio $\approx 1.893$. Thus, at least for memoryless
algorithms, knowing the exact deadlines helps.

As for upper bounds, we give the following deterministic
algorithms:
(i) A $1.737$-competitive algorithm for decremental queues (beating
the ratio $2$ of the na{\"i}ve greedy algorithm).
(ii) A $1.8$-competitive algorithm for FIFO queues.
(iii) A $\phi$-competitive algorithm for dynamic queues
when the item weights are non-decreasing (w.r.t.~their
position in the queue). This last result has implications
for packet scheduling, as all lower bound proofs
for that problem use instances where
packets' weights increase with deadlines.
Thus, either the ratio $\phi$ for the packet scheduling problem can
be achieved, or, to improve
the lower bound, one would have to use non-monotone instances
in the proof. Our competitive analysis uses a novel
invariant technique involving dominance relations
between sets of numbers, and is likely to find applications
in the analysis of packet scheduling.

Finally, we address the case
of dynamic queues and randomized algorithms. The
algorithm~{\Rmix} \cite{CCFJST06} can be easily
adapted to dynamic queues, achieving competitive ratio of $\e/(\e-1)$ against
an adaptive adversary. (In spite of the same value, this
result is not related to our analysis of the uniform dynamic sets.)
We prove a matching lower bound, which says that
no memoryless randomized algorithm for dynamic
queues can have competitive ratio smaller than $\e/(\e-1)$.

Due to space limitations some proofs are presented in the appendix.


\section{Preliminaries}

%


We refer to the items currently in $\calS$ as \emph{active}.
In other words, those are the items that have been already
inserted but not yet deleted.
An item is called \emph{pending} for an algorithm $\calA$ if it is active but not
yet collected by $\calA$. We denote the weight of item $x$
as $w_x$ and the total weight of a set of items $X$ as $w(X)$.

An instance of the problem is defined by a sequence $I$ of item
insertions or deletions. A solution consists of a \emph{selection
sequence} that specifies items selected at each step.
An optimal solution is computable in polynomial-time:
represent the instance as a bipartite graph $G$ whose
partitions are items and time steps. Item $a$ is connected
to the time steps when $a$ is active with edges of weight $w_a$.
The maximum-weight matching in $G$ represents an
optimal collection sequence.

When $\calS$ is a dynamic queue, $\calS$ is
represented by a list. We use symbol ``$\before$" to represent
the list ordering, i.e. $a \before b$ means that $a$ is
before $b$ in the list. In case of dynamic queues,
optimal solutions have special structure
that we explore in our competitiveness proofs, namely
they satisfy (w.l.o.g.) the following
\emph{Earliest-Expiration-First (EEF) Property:}
If $a,b$ are active at the same time, $a\before b$,
and both $a$, $b$ are collected, then $a$ is collected before $b$.

An online algorithm
$\calA$ is called \emph{$\calR$-competitive}, if for any instance $I$,
the gain of $\calA$ on $I$ is at least the optimum gain on $I$ divided
by $\calR$. (An additive constant is sometimes also
allowed in this bound; in all our upper bounds this constant is $0$, as
the weights can be scaled up.)
The \emph{competitive ratio} of $\calA$ is the smallest $\calR$
for which $\calA$ is $\calR$-competitive.


\section{Dynamic Sets}



For general dynamic sets,
it is easy to show the lower and upper bounds of $2$
on the competitive ratio of deterministic online algorithms.
To prove the lower bound, the adversary
can start with two items, say $a$
and $b$ with $w_a = w_b = 1$. Assume, by symmetry,
that the algorithm first collects $a$. Then the
adversary can collect $b$ in the first step, delete it, and
collect $a$ in the next step, while the algorithm has no
items to collect.

This bound can be achieved by a simple {\Greedy} algorithm:
at each step, if there is at least one pending item,
collect the maximum-value pending item.

\begin{fact}
{\Greedy} is $2$-competitive for dynamic sets.
\end{fact}

\begin{proof}
The proof is a straightforward adaptation of the proof for packet
scheduling \cite{Hajek01,KLMPSS04},
we include it for completeness sake.
The items collected by the adversary
are charged to the items collected by {\Greedy}.
Suppose the adversary collects an item $b$ at time $t$.
If this item was collected at time $t$ or earlier
by {\Greedy}, we charge it to that item in {\Greedy}'s sequence.
Otherwise, we charge it to the item collected
by {\Greedy} at time $t$.

Let $a$ be an item collected by {\Greedy} at some
time $u$. It receives at most two charges: one
from itself, if it was collected by the
adversary at time $u$ or later, and
the other one from the item $b$ collected by
the adversary at time $u$.
If $a$ receives a charge from $b$, then $b$
is pending for {\Greedy} at time $t$, and
therefore $w_b \le w_a$.
Therefore the charge to $a$ is at most $2w_a$.

Summarizing, all adversary items are charged
to our items, and each our item receives a charge
of at most twice its weight.
Thus {\Greedy} is $2$-competitive.
\qed
\end{proof}

Now we turn our attention to randomized algorithms.
For the adaptive adversary it is not hard to show a lower bound of $2$.
The adversary strategy is this: issue
$n$ items of weight $1$. Collect any item $a$ that
is collected by the algorithm with
probability at most $1/n$. Let $b$ be the
item collected by the algorithm.
If $b=a$, remove all items. If
$b\neq a$, remove all items except $b$ and
collect $b$. With probability at most
$1/n$, we have $b=a$ and both the algorithm and the adversary
collect $b$. With probability at least
$1-1/n$, the algorithm gets one item and the
adversary gets two items. So the ratio is
arbitrarily close to $2$.

For oblivious adversaries we concentrate on the uniform decremental case.
We show that for any randomized algorithm
the optimal strategy for the adversary is to repeatedly take and remove
an active item.
We consider {\UniRand} algorithm, which at each step collects
one of the pending items with equal probability.

For the lower bound we use the Yao min-max principle and consider inputs
which are constructed by the following random process: at each step choose
an active item uniformly at random, collect it, and delete it after the step.

Let $E_{a,p}$ be the expected number of items collected by an algorithm
if $a$~items are active and, among them, $p$ items are pending.
It can be shown that both in {\UniRand} analysis and in the lower bound construction,
$E_{a,p}$ satisfies the following recursive formula:
\begin{eqnarray*}
E_{a,0} = 0, \quad E_{a,1} = 1, \quad E_{a,p} &=&
	\frac{a-p+1}{a} \cdot E_{a-1,p-1} + \frac{p-1}{a} \cdot E_{a-1,p-2} + 1 \enspace.
\end{eqnarray*}
The unique solution of this formula can be estimated by $E_{a,p} \approx a(1-(1 - 1/a)^p)$.
Using exact approximations of $E_{a,p}$, we may prove the following theorems.

\begin{theorem}
\label{thm:sets-uniform-lower}
The competitive ratio of any randomized algorithm
for the uniform case of decremental sets is at least $\e/(\e-1)$.
\end{theorem}

\begin{theorem}
\label{thm:uni-rand}
Algorithm {\UniRand} is $\e/(\e-1)$-competitive for the uniform case of decremental sets.
\end{theorem}



\section{Dynamic Queues---Deterministic Algorithms}



In this section we consider deterministic online
algorithms for the case when $\calS$ is a
dynamic queue, that is, $\calS$ is an ordered list and only a prefix
of $\calS$ can be deleted.
In the fully dynamic case, items can be inserted anywhere,
while in the decremental case, all items are inserted at the beginning.

\begin{fact}
Every deterministic algorithm for
queues has competitive ratio at least
$\phi \approx 1.618$, even for decremental queues and with items
of ``life span" at most $2$.
\end{fact}

\begin{proof}
We start with two items in the list, $a \before b$,
with values $w_a = 1$ and $w_b = \phi$. If the algorithm
chooses $b$, the adversary chooses $a$ and deletes it,
and in the next step he chooses $b$, so the ratio is
$(w_a+w_b)/w_b = (1+\phi)/\phi = \phi$.
If the algorithm chooses $a$ in the first step,
the adversary chooses $b$ and
deletes both items, so the ratio is $w_b/w_a = \phi$ again.
\qed
\end{proof}


This section includes some lower and upper bounds for
the competitive ratio of deterministic online algorithms.
Starting with lower bounds,
we first prove that no deterministic algorithm for
the dynamic queue---in fact, even for the decremental
case---can achieve a competitive ratio better than $1.63$.
For memoryless algorithms we give a lower bound of $2$.

Our first upper bound concerns decremental queues, for
which we present a deterministic online algorithm
{\DecQueEFH} with competitive ratio $\approx 1.737$.
Next, we present a deterministic online algorithm
{\FIFOQueEH} that achieves competitive ratio $1.8$ for
FIFO Queues. We conclude this section with the algorithm
{\MarkAndPick} that
achieves competitive ratio $\phi$ for dynamic queues
in which the item weights are non-decreasing (this
also gives a $\phi$-competitive algorithm for
scheduling packets with non-decreasing weights).


\subsection{Lower Bounds for Decremental Queues}

We prove that no online deterministic algorithm can
have a competitive ratio smaller than $1.63$ for decremental queues.
The proof is by presenting an adversary's strategy that forces any
deterministic online algorithm $\calA$ to gain less than
$1/1.63$ times the adversary's gain.


\paragraph{Adversary's strategy.}
To get a cleaner analysis,
we first present the argument for the dynamic case (with insertions allowed),
and explain later how it extends to the decremental case. We assume that the
items appear gradually, so that at each step the algorithm has at most three
items to choose from.

Fix some $n\ge 2$. To simplify notation, in this section we refer to items simply by their weight, thus below ``$z_i$" denotes
both an item and its weight. The instance consists of a sequence of $2n$ items
$1,z_1,z_2,...,z_{2n-2},z_{2n}$ (note that
there is no item indexed $2n-1$) such that
\begin{equation*}
z_2 \before z_4 \before ... \before z_{2n-2}
		\before z_{2n} \before z_{2n-3} \before ... \before
		z_3 \before z_1 \before 1, \quad \textrm{ and }
\end{equation*}
\begin{equation*}
 1 > {z_1} > {z_2} > \ldots > z_{2n-3} > z_{2n-2} > z_{2n} > 0.
\end{equation*}
The even- and odd-numbered items in this sequence form two roughly
geometric sequences. In fact,
$z_{2i}$ is only slightly smaller than $z_{2i-1}$, for all $i=1,...,n-1$.

Initially, items $z_2 \before z_1 \before 1$ are present.
In step $i=1,2,...,n-1$, the adversary maintains the invariant that
the active items are $z_{2i}\before z_{2i-1} \before ... \before 1$, of which
only three items ${z_{2i}}$, ${z_{2i-1}}$ and $1$ are pending for $\calA$
(i.e. $\calA$ already collected $z_{2i-3},...,z_1$).
The adversary's move depends now on what $\calA$ collects in this step:
\begin{description}
	
    \item{(i)}
	$\calA$ collects ${z_{2i}}$. Then the adversary ends the game by
	deleting all active items. In this case
	the adversary collects $i$ heaviest items:
	$1,{z_1},{z_2},{z_3},\ldots,z_{i-1}$.

    \item{(ii)} $\calA$ collects $1$.
	The adversary ends the game by deleting $x_{2i}$ and $x_{2i-1}$.
	This leaves $\calA$ with no pending items, and the adversary
	can now collect $\calA$'s items one by one. Overall, in this case
	the adversary collects $2i$ heaviest items:
	$1, {z_1}, {z_2}, \ldots, {z_{2i-2}}, {z_{2i-1}}$.
	
	\item{(iii)} $\calA$ collects ${z_{2i-1}}$. In this case the game
	continues. If $i < n-1$, the adversary deletes $z_{2i}$, inserts
	$z_{2i+2}$ and $z_{2i+1}$ into the current list (according to
	the order defined earlier), and we go to step $i+1$.
	The case $i=n-1$ is slightly different: here the adversary
	only inserts the last item $z_{2n}$ before proceeding to step $n$
	(described below).
	
\end{description}

If the game reaches step $n$, $\calA$ has two pending items, $z_{2n}$ and $1$.
In this step, the adversary behavior is similar to previous steps:
if $\calA$ collects $z_{2n}$, then the adversary
deletes the whole sequence and collects $n$ heaviest items:
$1,{z_1},{z_2},{z_3},\ldots,z_{n-1}$. If $\calA$ collects $1$, the adversary
deletes $z_{n}$, leaving $\calA$ without pending items, and allowing
the adversary to collect the whole sequence.

Our goal is to find a sequence $\braced{z_i}$, as above,
and a constant $\calR$ such that
\begin{eqnarray}
\calR \cdot (1 + \textstyle{ \sum_{i=1}^j z_{2i-1} })
 	&\leq& 1 + \textstyle{ \sum_{i=1}^{2j+1} z_i } \qquad \text{for all}\; 0 \leq j < n
		\label{eqn: et queue lb, odd n}
 		\\
\calR \cdot (z_{2j+2} + \textstyle{ \sum_{i=1}^j z_{2i-1} })
	& \leq & 1 +  \textstyle{ \sum_{i=1}^j z_i } \qquad \text{for all}\; 0 \leq j < n
			\label{eqn: et queue lb, even n}
\end{eqnarray}

\begin{lemma}\label{lem: det queue lb, any n}
Suppose that there is a sequence $1,z_1,...,z_{2n-2},z_{2n}$, and
a constant $\calR$ that satisfy inequalities (\ref{eqn: et queue lb, odd n})
and (\ref{eqn: et queue lb, even n}). Then
there is no $\calR$-competitive deterministic online
algorithm for dynamic queues, even in the decremental case.
\end{lemma}

The lemma should be clear from the description of the strategy given
earlier, since the sums in inequalities  (\ref{eqn: et queue lb, odd n})
and (\ref{eqn: et queue lb, even n}) represent the gains of the
adversary and the algorithm in various steps. The only part that
needs justification is that the lemma holds in the decremental
case. To see this, we slightly modify the adversary
strategy: The sequence $\braced{z_i}$ is
created all at once in the beginning, and whenever $\calA$
deviates from the choices (i), (ii), (iii), it must be
collecting an item lighter than $z_{2i}$, and thus the
adversary can finish the game as in Case~(i).

\smallskip \noindent
Lemma~\ref{lem: det queue lb, any n} and straightforward calculations for
$n=3$ ($6$ items) yield the following.

\begin{theorem}\label{thm: det queue lower bound}
There is no deterministic online algorithm for dynamic
queues (even in the decremental case) with competitive
ratio smaller than $1.6329$.
\end{theorem}

A natural question arises how much this bound can be improved
with sequences $\braced{z_i}$ of arbitrary length. For $n=5$
($10$ items) one can obtain $\calR=1.6367...$ and our
experiments indicate that the corresponding ratios
tend to $\approx 1.6378458$, so the improvement is minor. However,
it is easy to prove a lower bound of $2$ for memoryless algorithms,
i.e. algorithms that decide which item to collect
based only on the weights of their pending items.


\begin{theorem}\label{thm: det queue lb memoryless}
For dynamic queues, no memoryless algorithm has competitive
ratio smaller than $2$.
\end{theorem}


\subsection{Upper Bound of $1.737$ for Decremental Queues}


\paragraph{Algorithm~{\DecQueEFH}:}
The computation is divided into stages, where each stage is a single
step, a pair of consecutive steps, or a triple of consecutive steps.
By $h$ we denote the maximum-weight pending item from the first
step of the stage. We use two parameters,
$\beta = (\sqrt{13}+1)/8$ and $\xi = (\sqrt{13}+1)/6$. Note that
$\beta < \xi$. Without
loss of generality, we assume that there are always pending
items, for we can always insert any number of
$0$-weight items into the instance, without changing the
competitive ratio. In the pseudo-code below, we assume that
after each item collection the algorithm proceeds to the
next step of the process.

\begin{center}
\begin{minipage}{3in}
\begin{tabbing}
\hspace{0.3in} \=\hspace{0.3in} \=\hspace{0.3in}\= \kill
(E) \> let $h$ be the heaviest pending item \\
	\> collect the earliest pending item $e$ with $w_e \ge \beta w_h$ \\
(F)	\> {\myif} $h$ is not pending {\mythen} end stage and goto (E) \\
	\> collect the earliest item $f$ with $w_f \ge \xi w_h$ \\
(H) \> {\myif} $h$ is not pending {\mythen} end stage and goto (E)\\
	\> collect $h$\\
	\> end stage and goto (E)
\end{tabbing}
\end{minipage}
\end{center}
%


\begin{theorem}\label{thm: static queue 1.73 algorithm}
For decremental queues, the competitive ratio of
{\DecQueEFH} is at most $\calR = 2(\sqrt{13}-1)/3\approx 1.737$.
\end{theorem}


\begin{proof}
We fix an instance and we compare {\DecQueEFH}'s gain on this instance
to the adversary's gain. We assume (w.l.o.g.) that	
the adversary has the EEF property.

The proof is by amortized analysis. We preserve
the invariant that, after each stage,
each item $i$ that is pending for the adversary
but has already been collected by {\DecQueEFH} has credit associated
with it of value equal $w_i$. The adversary's amortized gain is
equal to his actual gain plus the total credit change.
To prove the theorem, it is then sufficient to prove the following claim:

\smallskip\noindent
$(\ast)$ In each stage the adversary's amortized
gain is at most $\calR$ times {\DecQueEFH}'s gain.

\smallskip
To prove $(\ast)$, we consider several cases depending on the number of steps
in a stage and on relative location of items collected by the adversary
and by {\DecQueEFH}. We assume that at each step the adversary collects items
that were not collected by the algorithm before or during this stage.
Otherwise, either the adversary collects an item that has been collected by {\DecQueEFH} earlier and has credit on it, or {\DecQueEFH} collected this
item in this stage, in which case we can think of the algorithm giving
the adversary credit for this item anyway (even though it is not needed,
for this item is not pending for the adversary anymore).

We first observe that $e\beforeeq f\beforeeq h$ (if $f$ is defined for
this stage). This follows immediately from $\beta < \xi < 1$. An important
consequence of this, that plays a major role in the argument below,
is that when $h$ is deleted, then $e$ and $f$ are deleted as well.


\smallskip
\noindent
\mycase{1} $h$ is not pending in (F).
The stage has only one step, and the algorithm
collects an item $e$ with $w_e\ge \beta w_h$. Let
$a$ be the item collected
by the adversary.

If $h$ was deleted after (E) then $e$ was deleted as well.
Thus we do not give the adversary credit for $e$, and
his amortized gain is $w_a\le w_h$. The ratio is
\begin{eqnarray*}
\frac{w_a}{w_e} &\le& \frac{1}{\beta}\;=\; \calR.
\end{eqnarray*}
Suppose that $h$ was collected (i.e. $e=h$). If $a \before h$
then $w_a\le\beta w_h$, so, together with the credit for $h$,
the amortized adversary's gain is $w_a+w_h$, and
the ratio is
\begin{eqnarray*}
	 \frac{w_a + w_h}{w_h}
			&\le& 1+\beta
			\;\le\; \calR.
\end{eqnarray*}
If $h\before a$, we need not give the
adversary any credit, so the ratio is at most $1$.

\smallskip
\noindent
\mycase{2} $h$ is pending in (F) but is not pending in (H).
The stage has two steps, and we collect $e$ and $f$, gaining
$w_e + w_f \ge (\beta+\xi)w_h$.
The adversary collects two items, say $a$ and $b$ with $a\before b$.

If $h$ was deleted after (F), then both $e$ and $f$ are deleted as well,
so we do not give the adversary any credits. Thus the adversary
amortized gain is $w_a+w_b\le 2 w_h$. In this case the ratio is
\begin{eqnarray*}
\frac{w_a + w_b}{w_e+w_f}
	&\le&  \frac{2}{\beta +\xi}
	\;=\; 4(\sqrt{13}-1)/7
				\;<\; \calR.
\end{eqnarray*}
Suppose now that
the algorithm collected $h$ in (F) (i.e. $f=h$), and thus our
gain is actually $w_e+w_h \ge (1+\beta)w_h$. If $b\before e$,
then the adversary amortized gain is
$w_a + w_b + w_e + w_h \le (2\beta+1) w_h + w_e$. Thus the ratio is
\begin{eqnarray*}
\frac{w_a + w_b + w_e + w_h}{w_e+w_h}
	&\le&  \frac{3\beta+1}{\beta +1}
	\;=\; 4(\sqrt{13}+10)/17
				\;<\; \calR.
\end{eqnarray*}
If $e\before b\before h$, the adversary's amortized gain is
$w_a + w_b + w_h \le (2\xi +1)w_h$, so the ratio is
\begin{eqnarray*}
\frac{w_a + w_b + w_h}{w_e+w_h}
	&\le&  \frac{2\xi+1}{\beta +1}
	\;=\; 2(5\sqrt{13}+23)/51
				\;<\; \calR.
\end{eqnarray*}
If $h\before b$, we need not give the adversary any
credit, gaining $w_a + w_b \le 2w_h$.
So the ratio is less than the one above because $2\xi + 1 > 2$.

\smallskip
\noindent
\mycase{3}  $h$ is still pending in (H). In this case the stage
has three steps and we collect $e$, $f$, and $h$,
for the total gain of $w_e + w_f + w_h \ge (\beta+\xi+1)w_h$.
The adversary collects three items $a\before b \before c$ and
may get credit for some items.

If $c\before e$, then the adversary gets credit for $e$, $f$ and
$h$, and his amortized gain is $w_a + w_b + w_c + w_e + w_f + w_h
	\le 3 \beta w_h + w_e + w_f + w_h$.
So the ratio is
\begin{eqnarray*}
\frac{w_a+w_b+w_c+w_e+w_f+w_h}{w_e + w_f+w_h}
			&\le& \frac{4\beta+\xi+1}{\beta+\xi+1}
			\;=\; \calR.
\end{eqnarray*}
If $e\before c \before f$, then the adversary
 gets credit for $f$, $h$, and his amortized gain is
$w_a + w_b + w_c +w_f + w_h \le 3\xi w_h + w_f + w_h$. So the ratio is
\begin{eqnarray*}
\frac{w_a+w_b+w_c+w_f+w_h}{w_e + w_f+w_h}
			&\le& \frac{4\xi+1}{\beta+\xi+1}
			\;=\; \calR.
\end{eqnarray*}
If $f\before c\before h$, then he only
gets credit for $h$, so his amortized gain is
$w_a + w_b + w_c + w_h \le 4 w_h$. The ratio is
\begin{eqnarray*}
\frac{w_a+w_b+w_c+w_h}{w_e + w_f+w_h}
		 &\le& \frac{4}{\beta+\xi+1}
			\;=\; 8(31 - 7\sqrt{13})/27
		\;<\; \calR.
\end{eqnarray*}
Finally, when $h\before c$, the adversary amortized
gain is $w_a + w_b + w_c \le 3 w_h$, which is less
than in the previous case.
\qed
\end{proof}

We attempted to extend the idea of Algorithm~{\DecQueEFH} to stages
with more steps, but according to numerical experiments we conducted,
this does not improve the competitive ratio.


\subsection{Upper Bound of $1.8$ for FIFO Queues}

We now extend the idea of the previous algorithm
to FIFO queues.
The algorithm uses two parameters, $\alpha$ and $\beta$.
The main idea is this: If the heaviest item $h'$ from the previous
step is no
longer pending, this is fine and it simply begins another step. If
$h'$ is still pending and new items have been inserted to the queue, the
algorithm inspects them. If the heaviest new item $h$ is not too heavy
(i.e., if $\alpha w_{h} \leq w_{h'}$), the algorithm ignores new items
and collects $h'$. If $h$ is very
heavy ($\alpha w_{h} > w_{h'}$), the algorithm forgets about $h'$, resets
\emph{heavy} to $h$ and collects the earliest pending item $e$, such that
$w_e \geq \beta w_{h}$.

\paragraph{Algorithm~{\FIFOQueEH}:}
By $h$ we denote the maximum-weight pending item and by $h'$ the previous
maximum-weight pending item (initially $h'$ is an imaginary item of
weight $0$). We use two parameters, $0 < \alpha , \beta  < 1 $. Without loss
of generality, we assume that there are always pending items, for we can
always insert any number of $0$-weight items into the instance,
without changing the competitive ratio. In the pseudo-code below, we assume
that after each item collection the algorithm proceeds to the next step of the
process.

\begin{center}
\begin{minipage}{3in}
\begin{tabbing}
\hspace{0.3in} \=\hspace{0.3in} \=\hspace{0.3in}\= \kill
  \> let $h$ be the heaviest pending item and $h'$ the previous heaviest item \\
(E)	\>{\myif} ($h'$ is not pending) or ($h'$ is pending and $\alpha w_h \geq w_{h'}$) {\mythen} \\
	\> collect the earliest item $e$ with $w_e \geq \beta w_h$ \\
(H) \> {\myelse} collect $h'$
\end{tabbing}
\end{minipage}
\end{center}

\begin{theorem}\label{thm: fifo det 1.8 algorithm}
The competitive ratio of {\FIFOQueEH} with $\alpha= \frac{3}{4}$ and $\beta=\frac{2}{3}$
is at most $1.8$.
\end{theorem}

By inspecting the decremental queue instance
$a \before b \before c \before d$, with weights
$w_a=w_b=\beta - \epsilon$, $w_c=\beta$, $w_d=1$, and weights
$w_a\beta$, $w_b=w_c=1 - \epsilon$, $w_d=1$, one can
observe the following.

\begin{theorem}\label{thm: tight analysis of FIFOQueEH}
The competitive ratio of Algorithm~{\FIFOQueEH} is at least $1.8$,
regardless of the choice of $\alpha$ and~$\beta$.
\end{theorem}

\subsection{Upper Bound $\phi$ for Non-Decreasing Weights}

In this section we give an online algorithm {\MarkAndPick}
that is $\phi$-competitive for dynamic queues when item
weights are increasing. More precisely, if $\before$
denotes the ordering of the items in the queue, then we
assume that, at any time, for any two active items $a,b\in \calS$,
if $a\before b$ then $w_a\le w_b$.


\paragraph{Algorithm~{\MarkAndPick}:}
\begin{center}
\begin{minipage}{2in}
\begin{tabbing}
\hspace{0.2in} \=\hspace{0.2in} \=\hspace{0.2in}\= \kill
  \>At each step $t=1,2,...$ \\
	\>{\myif} there is no pending item {\mythen} wait, {\myelse} \\
	\> let $h$ be the heaviest unmarked item (not necessarily active) \\
	\> mark $h$, collect the earliest pending item $i$ with $w_i \ge w_{h}/\phi$ \\
\end{tabbing}
\end{minipage}
\end{center}



\emph{Basic idea:}
The proof is based on a charging scheme, where the items collected
by the adversary are charged to our items in such a way that each
our item is charged at most $\phi$ times its weight. In other words,
each our item $i$ has a budget $\phi w_i$ and it uses its budget
to pay for some items in the adversary's set. In fact, each such $i$
pays for either one or two adversary's items.
The charging is done in two steps: (1)~we charge the adversary items to the marked items, and (2)~we charge the marked items to the algorithm's items.

When we mark an item $h$, we collect an item $i$ with $w_i\ge w_h/\phi$,
so its budget $\phi w_i$ is sufficient to pay for $h$. In a simple
scenario, if we collect items in all steps, we can afford to pay for all
items marked in step (2). In step (1), we also show that
the weight of the marked items exceeds that of the adversary,
and these two facts easily imply $\phi$-competitiveness.

In reality, the situation is more complicated---in some
moves {\MarkAndPick}
does not have any items to collect. For instance,
suppose that the adversary collects an item $j'$ in such a step.
This item is already collected by the algorithm and
is now marked. Its budget is $\phi w_{j'}$, and it pays
for its mark, as well as for another item collected by
the adversary in the past. Roughly, this is the
item collected by the adversary when the algorithm was marking $j'$.

The principle is, that any idle step
is a consequence of a step in which the algorithm marked
$h=j'$, collected an item $i$, and the
adversary collected an item $j$ smaller than $i$ and
still pending for the algorithm. (This is not exactly
correct, but it reflects the main principle. It may happen
that the marked item $h$ "responsible" for the idle
step in which the adversary collects $j'$ is actually
different from $j'$, but later in the process $h$
"transfers" this responsibility to $j'$.)
In that case, $w_j\le w_{j'}/\phi$. We refer to such
$j$'s as "extra" items (even though those are not
exactly the extra moves, but they cause extra moves later).
This idea can be formalised and proved rigorously.

\begin{theorem}\label{thm: det queue incr weight}
Algorithm~{\MarkAndPick} is $\phi$-competitive for
dynamic queues if item weights are non-decreasing.
\end{theorem}


\section{Dynamic Queues---Randomized Algorithms}



In this section, we consider randomized algorithms for dynamic queues.
Chin~{\etal}~\cite{CCFJST06} designed a randomized
algorithm {\Rmix} for bounded-delay packet scheduling.
This algorithm is memoryless and achieves competitive
ratio $\e/(\e-1)$ against an adaptive
adversary. The competitiveness proof for {\Rmix}
applies, with virtually no changes, to dynamic queues.
In this section we show that this bound is tight.


\begin{theorem}\label{thm: dynamic queue rand m-less lower bound}
For dynamic queues, any memoryless randomized algorithm has a competitive
ratio at least $\frac{\e}{\e-1}$ against an adaptive-online adversary.
\end{theorem}

\noindent
{\em Basic idea.}
We fix an online memoryless randomized algorithm $\calA$. At the beginning
$\calA$ is given $n+1$ items: $a^0,a^1,\ldots,a^n$, where
$a = 1 + \frac{1}{n}$. We consider $n+1$ strategies for an adversary.
The \mbox{$k$-th} ($0 \leq k \leq n$) strategy is as follows: in each step
collect $a^k$, delete items $a^0,a^1,...,a^k$, and issue their new copies.
If $\calA$ collected $a^j$ for $j>k$, issue a copy of $a^j$ as well.
This way, in each step exactly one copy of each item is pending for $\calA$,
while the adversary accumulates copies of the items $a^j$ for $j>k$.
Since $\calA$ is memoryless, in each step it uses the same probability
distribution $(q_j)_{j=0}^n$, where $q_j$ is the probability of collecting
$a^j$.

This step is repeated $T \gg n$ times, and after the last step
both the adversary and $\calA$ collect all their pending items.
Since $T \gg n$, we may focus only on the
expected amortized profit in a single step, which is
$a^k + \sum_{i>k} q_i a^i$ for the adversary and
$\sum_i q_i a^i$ for $\calA$. By solving a set of linear equations,
we show (non-constructively) that for any probability distribution
$(q_j)_j$, there exists $k$, for which the ratio between these
gains is at least $\e / (\e-1)$.


\section{Final Comments}



We provided upper and lower bounds for the competitive ratio
of several variants of the item selection problem.
The most intriguing open problems are
to establish better bounds for (1) \emph{dynamic queue case}
and nontrivial bounds for (2) \emph{general randomized case.}
For (1), we have shown a lower bound
of $\approx 1.63$ (improving the lower bound of $\phi$),
but no upper bound better than $2$ is known.
(Better bounds are known for packet scheduling
\cite{LiSeSt07,EngWes07}, but these algorithm use
information about packet deadlines and do not seem to
apply to dynamic queues.) We have also shown
better bounds for some restricted cases:
$1.8$~for the FIFO queue (a generalization
of packet scheduling with agreeable deadlines) and
$\approx 1.737$ for the decremental queue.
For (2), we have an $\e/(\e-1)$ lower bound,
and we have shown it is tight in the uniform
decremental case.


{\small
\bibliographystyle{abbrv}
\bibliography{online,other}
}

\newpage


\setcounter{secnumdepth}{3}

\begin{appendix}

\section{The Uniform Case of Decremental Sets: Randomized Algorithms against Oblivious Adversaries}


In this part of appendix, we prove Theorems~\ref{thm:sets-uniform-lower}
and~\ref{thm:uni-rand}, i.e.~we show that for the uniform case of decremental sets:
\begin{enumerate}[(i)]
\item the competitive ratio of any randomized algorithm is at least $\e/(\e-1)$,
\item the algorithm {\UniRand} is $\e/(\e-1)$-competitive.
\end{enumerate}
Obviously, the lower bound holds also for non-uniform cases as well.


\subsection{Lower Bound}

Fix any set size $n$ and let $\calR = \frac{n}{n+1} \cdot \frac{\e}{\e-1}$.
The proof is based on Yao's min-max principle, 
i.e. we construct a probability distribution $\pi$ over inputs, such that
(i) the gain of an optimal solution is $n$ on any input sequence from the support of $\pi$,
and
(ii) the expected gain of any deterministic algorithm run on a~randomly chosen
input from $\pi$ is at most	$n / \calR$.

We construct the distribution $\pi$ implicitly by the following random
process: at each step, choose uniformly at random any active item $a$,
collect it, and delete it after the step. Obviously, property (i) holds,
and it is sufficient now to show~(ii).

Fix any deterministic algorithm $\calA$. Without loss of generality,
$\calA$ always collects an item if one is pending.
Let $E_{a,p}$ be the expected number of items collected by $\calA$
if $a$ items are active and, among them, $p$ items are pending.
If $p =0,1$, then $\calA$ collects $p$ items.
On the other hand, if $p > 1$, $\calA$ collects an item in the
first step, reducing the number of pending items to $p-1$, and then,
with probability $(p-1)/a$ another pending element may get deleted.
Hence, $E_{a,p}$ satisfies the following recurrence for any $a\ge 1$:
$E_{a,0} = 0$, $E_{a,1} = 1$, and
\begin{eqnarray*}
E_{a,p} &=&
	\frac{a-p+1}{a} \cdot E_{a-1,p-1} + \frac{p-1}{a} \cdot E_{a-1,p-2} + 1
\end{eqnarray*}
for $a\ge p \ge 2$.

The following technical lemma, can be used to bound the performance of any
deterministic algorithm on $\pi$.


\begin{lemma}\label{lem: bound on E_ap}
For all $a\ge p \ge 0$, we have $E_{a,p} \leq (a+1) (1-\e^{-p/a}) + 1$.
\end{lemma}

\begin{proof}
We prove the lemma by induction on $p$. The case $p=0$ holds trivially.
For $a\ge p = 1$, $E_{a,p}=1 \leq (a+1)(1 - \e^{-1/a}) + 1$, as
$(a+1)(1 - \e^{-1/a}) > 0$.

For $a\ge p\ge 2$, we use the inductive assumption and
the recursive definition of $E_{a,p}$, getting
\begin{eqnarray*}
E_{a,p} &= &
		\frac{a-p+1}{a} \cdot E_{a-1,p-1} + \frac{p-1}{a} \cdot E_{a-1,p-2} + 1 \\
	&\leq &
	 	(a-p+1)\left(1 - \e^{ - \frac{p-1}{a-1} } \right)
		+ (p-1) a \left(1 - \e^{ - \frac{p-2}{a-1} } \right) + 2
		\\
	&= & a+2
			- \e^{-p/a}\cdot
				\brackd{ (a-p+1)\e^{ \frac{a-p}{a(a-1)} }
					+ (p-1)\e^{ \frac{2a-p}{a(a-1)} }
					}
				\\
	&\leq & a+2
		 - \e^{-p/a} \cdot \brackd{
		 			(a-p+1) \left(1+ \frac{a-p}{a(a-1)}\right)
						+ (p-1) \left(1 + \frac{2a-p}{a(a-1)} \right)
						}
		\\
		&= & (a+1) (1-\e^{-p/a}) + 1\enspace,
\end{eqnarray*}
where the second inequality follows from $\e^x \ge 1+x$. This
completes the inductive step and the proof of the lemma.
\qed
\end{proof}

By Lemma~\ref{lem: bound on E_ap}, we have that $E_{n,n}$, the expected number
of items collected by~$\calA$, is at most $(n+1)(\e-1)/\e + 1 = n/\calR+O(1)$.
Thus, property (ii) holds and, applying Yao's principle and taking the
limit on $n$, we obtain our lower bound.


\subsection{Upper Bound}

In this part we show that the algorithm {\UniRand} (which at each step collects one
randomly chosen pending item) is $\e/(\e-1)$-competitive.

The idea behind the competitive analysis is to
prove first that the optimal strategy of the adversary can be easily described:
without loss of generality, at each step the adversary
collects one item and deletes the same item. Then we analyze the
competitive ratio of {\UniRand} against such an adversary.

\medskip

The number of items in the initial set is denoted by $n$. By
$a$ we denote the number of active items at a given step,
and by $p\le a$ the number of pending items. Note that
$p$ is a random variable. We consider states of the game
between {\UniRand} and the adversary,
conditioned on $p$ being fixed, and we refer to such a state
as a \emph{configuration $(a,p)$}.  The definition
of {\UniRand} implies that in configuration $(a,p)$
each active item is pending with  probability $p/a$.

Starting from each fixed configuration $(a,p)$,
we analyze the relation between the gain of {\UniRand} and
that of the adversary.
Since the items are identical
and {\UniRand}'s pending items are distributed
uniformly, we need not specify which items are
deleted by the adversary in a given step, only their
number.

We can split these steps into smaller parts: an \emph{elementary}
step consists in either collecting an item or deleting an item.
We can thus describe the adversary's strategy as a sequence
$S \in \{\take, \dele \}^*$, where $\take$ means that the adversary
collects an item (and allows {\UniRand} to collect one as well,
provided it has some pending items)
and $\dele$ means that the adversary deletes an item.
In the following we consider
only the \emph{feasible} strategies~$S$ for configuration $(a,p)$,
i.e. those in which:
(i) the number of $\dele$~operations in~$S$ is $a$,
(ii) in every suffix of $S$ the number of $\take$ operations does
not exceed the number of $\dele$~operations.
We call strategy~$S$ a $k$-{\em strategy} if it contains $k$ $\take$
operations, i.e.~the gain of the adversary on $S$ is $k$.

By $E_{a,p}[S]$ we denote the expected gain of
Algorithm~{\UniRand} starting from configuration $(a,p)$,
versus adversary strategy $S$.

The following lemma shows that the best strategy for the adversary, which guarantees
that the adversary takes $k$ elements, is $(\take\dele)^k$ run on $k$-element set.

\begin{lemma}\label{lemma:k_optimal_strategy}
Let $S$ be a $k$-strategy for $(a,a)$. Then
$E_{a,a}[S] \geq E_{k,k}[(\take\dele)^k]$.
\end{lemma}

On the other hand for these simple adversary strategies, we may appropriately bound
the $E_{a,p}$ values.

\begin{lemma}\label{lemma:upper_bound_estimation}
For any integers $p \leq a$, it holds that
$E_{a,p}[(\take\dele)^a] \geq a(1-(1 - 1/a)^p)$.
\end{lemma}

While we prove the lemmas above in the following subsections, we argue now
how they imply the competitiveness of {\UniRand}.

\begin{proof}[of Theorem~\ref{thm:uni-rand}]
We fix any $n$-element set and any adversary's strategy $S$. Let $k$ be the number of
items collected by the adversary.
By Lemmas~\ref{lemma:k_optimal_strategy} and \ref{lemma:upper_bound_estimation},
the gain of {\UniRand} is $E_{n,n}[S] \geq E_{k,k}[(\take\dele)^k] \geq k (1 - (1-1/k)^k)$.
Therefore, the competitive ratio is at most
$$
\frac k {E_{n,n}[S]}  \;\leq\;
  		\brackd{1 - \left(1 - \frac {1}{k}\right)^k }^{-1}
			\;<\;   \frac {\e} {\e-1} \enspace.
$$
\qed
\end{proof}


\subsubsection{Relations between Adversary Strategies\\\\}%
\medskip%
In this subsection, we show why $(\take\dele)^k$ is the best $k$-strategy. First,
we prove two technical lemmas.

\begin{lemma}\label{lem: E_ap monotone with p}
For any $p < a$ and strategy $S$, it holds that $E_{a,p+1}[S] \geq E_{a,p}[S]$.
\end{lemma}

\begin{proof}
We prove the inequality by induction on pairs $(p,S)$,
where $(p_1,S_1)<(p_2,S_2)$ if $|S_1| < |S_2|$ or
$|S_1| = |S_2|$ and $p_1<p_2$.
The inductive basis is straightforward:
if $S=\dele^a$ then $E_{a,p+1}[\dele^a] = 0 = E_{a,p}[\dele^a]$.
If $p=0$, then clearly $E_{a,1}[S] \geq 0 = E_{a,0}[S]$.

In the inductive step, we have two cases depending on
the first action of $S$. If $S=\take S'$, then
$E_{a,p+1}[\take S']=1+E_{a,p}[S'] \geq 1+E_{a,p-1}[S']=E_{a,p}[\take S']$,
by the inductive assumption.
If $S=\dele S'$, since we have $p+1$ pending items,
the adversary deletes a pending item
with probability $\frac{p+1}{a}$ and a non-pending item with
probability $\frac{a-p-1}{a}$. Using the inductive assumption, we get
\begin{align*}
E_{a,p+1}[\dele S']
	=&\; 	\frac{a-p-1}{a} E_{a-1,p+1}[S'] + \frac{p+1}{a} E_{a-1,p}[S'] \\
	\geq&\; \frac{a-p-1}{a} E_{a-1,p}[S'] +
			\brackd{ \frac{1}{a} E_{a-1,p}[S'] + \frac {p}{a}E_{a-1,p-1}[S'] } \\
	=& \; E_{a,p}[\dele S'] \enspace.
\end{align*}
\qed
\end{proof}

\begin{lemma}\label{lemma:td_inversing}
Let $S=S_1\take \dele S_2$ be a feasible strategy for $(a,p)$.
If $S'=S_1\dele \take S_2$ is also a feasible strategy, then
$E_{a,p}[S] \geq E_{a,p}[S']$.
\end{lemma}

\begin{proof}
Suppose first that $S=\take \dele S_2$.
If $p=0$, then the claim is obvious, as $E_{a,0}[S]=E_{a,0}[S']=0$, and
if $p=1$, then $E_{a,1}[S]=1 \geq E_{a,1}[S']$.
So consider $p>1$. By direct calculation and Lemma~\ref{lem: E_ap monotone with p}, we get
\begin{eqnarray*}
E_{a,p}[\take \dele S_2]
	&=&  1+E_{a,p-1}[\dele S_2]
			\\
	&=&  1+\frac{a-p+1}{a} E_{a-1,p-1}[S_2]+ \frac{p-1}{a} E_{a-1,p-2}[S_2]
			\\
	&\geq& 1+\frac{a-p}{a} E_{a-1,p-1}[S_2]+ \frac{p}{a} E_{a-1,p-2}[S_2]
	 		\\
	&=& \frac{a-p}{a} E_{a-1,p}[\take S_2]+ \frac{p}{a} E_{a-1,p-1}[\take S_2]
			\\
	&=& E_{a,p}[\dele\take S_2] \enspace.
\end{eqnarray*}
Now consider the general case, when $S=S_1\take\dele S_2$ and $S'=S_1\dele\take S_2$. Then
\begin{align*}
E_{a,p}[S] \;=\; E_{a,p}[S_1\take\dele S_2]
			\;=\; \sum_i c_i E_{a-k,p-i}[\take\dele S_2]
	\\
E_{a,p}[S'] \;=\; E_{a,p}[S_1\dele \take S_2]
		\;=\; \sum_i c_i E_{a-k,p-i}[\dele\take S_2]
\end{align*}
for some $k$, $\{c_i\}$. By the previous calculations, it holds that
$E_{a-k,p-i}[\take\dele S_2] \geq E_{a-k,p-i}[\dele\take S_2]$ for all $i$,
and hence
\begin{align*}
E_{a,p}[S] =&\; E_{a,p}[S_1\take\dele S_2] \\
			=&\; \sum_i c_i E_{a-k,p-i}[\take\dele S_2] \\
			\geq&\; \sum_i c_i E_{a-k,p-i}[\dele\take S_2] \\
			=&\; E_{a,p}[S_1\dele\take S_2] \\
			=&\; E_{a,p}[S'] \enspace,
\end{align*}
completing the proof.
\qed
\end{proof}

\begin{proof}[of Lemma~\ref{lemma:k_optimal_strategy}]
We take any feasible $k$-strategy $S$ starting from configuration $(a,a)$. It is straightforward, that
we may iteratively apply Lemma~\ref{lemma:td_inversing} and swap
consecutive $\take$ and $\dele$ operations, obtaining
a feasible strategy $\dele^{a-k}(\take\dele)^k$ at the end.
Thus $E_{a,a}[S] \geq E_{a,a}[\dele^{a-k}(\take\dele)^k] = E_{k,k}[(\take\dele)^k]$
\qed
\end{proof}


\subsubsection{Bound for the Best Adversary Strategy}

\begin{proof}[of Lemma~\ref{lemma:upper_bound_estimation}]
For $p = 0,1$, $E_{a,p}[(\take\dele)^a] = p$ and the lemma trivially holds.
Also if $a = p = 2$, then
$E_{a,p}[(\take\dele)^a] = 3/2 = a(1-(1 -1/a)^p)$ .

For the remaining values of $a$ and $p$ we prove the lemma by induction on $a$.
We note that for $p \geq 2$
$$
E_{a,p}[(\take\dele)^a] \;=\;
 		\frac{a-p+1}{a} E_{a-1,p-1}[(\take\dele)^{a-1}]
				+ \frac{p-1}{a} E_{a-1,p-2}[(\take\dele)^{a-1}] + 1 \enspace,
$$
as the pending items are distributed uniformly among active items.
Note that this is the same recurrence as in the proof of the lower bound.
Thus, using the induction assumption, we get
\begin{eqnarray*}
	E_{a,p}[(\take\dele)^a]
		&\geq
			&\frac{a-p+1}{a} \left(a-1\right)\left(1-\left(\frac {a-2}{a-1}\right)^{p-1}\right) \\
			&&\enspace \;+\; \frac{p-1}{a} \left(a-1\right)\left(1-\left(\frac {a-2}{a-1}\right)^{p-2}\right)
			+ 1
			\\
	&=&
		a - \frac{(a-1)^2}{a} \left(\frac {a-2}{a-1}\right)^{p-2}
				\left(1 + \frac{p-2}{(a-1)^2} \right)
			\\
	&\geq&
		a - \frac{(a-1)^2 }{a}\left(\frac {a-2}{a-1}\right)^{p-2}
			\left(1 + \frac{1}{(a-1)^2} \right)^{p-2}
		 \\
	&\geq&
	 	a - \frac{(a-1)^2}{a} \left(\frac {a-1}{a}\right)^{p-2}
		\\
	&=&  a \left(1 - \left(\frac {a-1} a \right)^p\right),
\end{eqnarray*}
completing the proof.
\qed
\end{proof}

\section{Deterministic Algorithms for Dynamic Queues}
\subsection{Lower Bounds}

\begin{proof}[of Theorem~\ref{thm: det queue lower bound}]
We now exhibit a sequence of $6$ items for which
Lemma~\ref{lem: det queue lb, any n} holds with $\calR \approx 1.63$.
To simplify notation, we rename the items:
$z_2, z_4, z_6, z_3, z_1$ as  $x,y,z,u,v$, respectively.
Otherwise we follow the aforementioned idea.

By Lemma~\ref{lem: det queue lb, any n}, we want to find
numbers $x,y,z,u,v$ such that $0 < z < y < u < x < v < 1$ and
a maximal $\calR$ for which:
\begin{eqnarray*}
\calR \cdot x & \leq & 1 \\
\calR \cdot 1 & \leq & 1+v \\
\calR \cdot (v+y) & \leq & 1+v \\
\calR \cdot (1+v) & \leq & 1+v+x+u \\
\calR \cdot (v+u+z) & \leq & 1+v+x \\
\calR \cdot (1+v+u) & \leq & 1+v+x+u+z
\end{eqnarray*}
We can solve it by replacing inequalities by equations, and
after doing substitutions, the problem reduces to finding
a solution of a polynomial equation $ x^5 + x^4 + 5x^3 - x^2 - 1 =0$.
This polynomial has exactly one real root, $x = 0.61238...$,
which yields $\calR = 1.6329...$.\qed
\end{proof}

\begin{proof}[of Theorem \ref{thm: det queue lb memoryless}]
Fix a memoryless algorithm $\calA$.
We give an adversary's strategy where the adversary's gain
is $2-o(1)$ times $\calA$'s gain.

Pick large integers $n$ and $T\gg n$, and
let $X = \braced{x_0,...,x_n}$ where
$w_{x_i} = 1+\frac{i}{n}$ for $i = 0,1,...,n$.
The adversary maintains the invariant that
at each step $\calA$'s pending set is $X$, with the items
ordered by increasing value. Suppose that
for this pending set $\calA$ collects some item $x_k$.

If $k=0$, the adversary collects item $x_n$,
deletes all items, inserts copies of all items from $X$ again
into the queue, and repeats the process $T$ times.
$\calA$'s gain is $Tw_{x_0} = T$ while the optimum gain is $Tw_{x_n} = 2T$,
so the ratio is $2$.

Suppose now that $k \ge 1$. In this case,
the adversary collects $x_{k-1}$, deletes all items $x_0,...,x_{k-1}$ for
$i=0,1,...,k-1$ and inserts new copies of
items $x_0,...,x_k$. This process is repeated $T$ times.
After $T$ steps, the adversary
collects the remaining uncollected items, in particular,
all $T$ copies of item $x_k$. $\calA$ can of course
collect the remaining pending items.
The value collected by $\calA$
is at most $Tw_{x_k}+2(n+1) = T(1+k/n) + 2(n+1)$,
while the value collected by the adversary is at least
$T(w_{x_{k-1}} + w_{x_k}) = T(2+(2k-1)/n)$. So with
$T = n^3$ and $n \to\infty$ the ratio approaches $2$.
\qed
\end{proof}

\subsection{Upper Bound for FIFO Queues}

\begin{proof}[of Theorem~\ref{thm: fifo det 1.8 algorithm}]

We define \emph{stages} of the algorithm. The first stage begins before {\FIFOQueEH}
collects any item. Each next stage begins immediately after previous stage
ends. The stage ends when $h$ is deleted by the adversary or
condition (H) holds. The last step of the stage is the last step $t$
such that at beginning of $t$, $h'$ was pending for {\FIFOQueEH}.

Let $e_1,e_2, \ldots e_k$ be the set of items collected (in this order) by {\FIFOQueEH}
in one stage when (E) condition held.
Let $h_1, \ldots, h_k$ be the corresponding heavy items. From the algorithm and
the definition of FIFO queues, we have:

\begin{fact}
For all $i=1,...,k$ we have $e_i \beforeeq h_i$ (with all relations strict,
except possibly for $i=k$) and $h_i\beforeeq h_{i+1}$.
\end{fact}

The proof is by amortized analysis. We preserve
the invariant that, after each stage,
each item $i$ that is pending for the adversary
but has already been collected by {\FIFOQueEH} has credit associated
with it of value equal $w_i$. The adversary's amortized gain is
equal the his actual gain plus the total credit change.
To prove the theorem, it is then sufficient to prove the following claim:

\smallskip\noindent
$(\ast)$ In each stage the adversary's amortized
gain is at most $1.8$ times {\FIFOQueEH}'s gain.

\smallskip
We prove $(\ast)$ for each of the following cases
(1)
the stage ends because the heaviest item is deleted,
or
(2)
the stage ends because {\FIFOQueEH} collects $h'$.
The second case has three sub-cases, depending
on which condition the latest item $a$ collected by the
adversary satisfies:
(2a) $a \before e_l$ for some $l \leq k$,
(2b) $e_k \beforeeq a \before h_k$, or
(2c) $h_k \beforeeq a $.

\noindent
\mycase{1} The stage ended because item $h$ was deleted.
The adversary cannot gain credit for any items taken
by the algorithm, as they are not pending after that stage.
Hence the gain of the adversary is at most
$\sum_{i=1}^k w_{h_i}$. The gain of the algorithm is
$\sum_{i=1}^k w_{e_i}$.
As $w_{e_i} \geq \beta w_{h_i}$ for all $i$,
the competitive ratio is at most~$\frac 1 \beta = 1.5$.

\noindent \mycase{2}
The stage ends, because $\alpha w_{h_{k+1}} < w_{h_k}$ and {\FIFOQueEH}
collected $h_k$.

\noindent \mycase{2a}
Suppose that $a \before e_l $ for some $l \leq k$, choose minimal such $l$.
Let $c_i = w_{e_i} - \beta w_{h_i}$ for $i=1, \ldots , l-1$.
We estimate the amortized  adversary's gain:
in step $i \leq l-1$ the heaviest pending item is $h_i $, and
in step $i \in \{l \ldots  k+1\}$ the adversary collects item preceding $e_l$.
If this item is still pending for the algorithm,
its weight is at most $\beta w_{h_l}$.
If it is not pending for the algorithm, it is one of the $e_i$'s,
for some $i<l$. Then the gain of the adversary is
$w_{e_i} = \beta w_{h_i} + c_i < \beta w_{h_l} + c_i $.
As we can collect each such $e_i$ only once, all
these steps add to at most $(k-l+2) \beta w_{h_l} + \sum_{i=1}^{l-1} c_i$.
The adversary gets credit for all items taken by
the algorithm in steps $l,\ldots, k+1$ that is
$\sum_{i=l}^k w_{e_i} + w_{h_k}$.
Thus the amortized gain of the adversary is at most
$$
\sum_{i=1}^{l-1} w_{h_i} + (k-l+2)\beta w_{h_l} + \sum_{i=1}^{l-1} c_i + \sum_{i=l}^k w_{e_i} + w_{h_k}\;,
$$
On the other hand the gain of the algorithm is
$$
\sum_{i=1}^k w_{e_i} + w_{h_k} = \sum_{i=1}^{l-1}(\beta w_{h_i} + c_i ) + \sum_{i=l}^k w_{e_i} + w_{h_k}\;.
$$
Hence for fixed $l$, the competitive ratio of {\FIFOQueEH} in the stage is at most
$$
\calR_{1,l} = \frac{\sum_{i=1}^{l-1} h_i  + (k-l+2) \beta h_l + \sum_{i=1}^{l-1} c_i +
\sum_{i=l}^{k} e_i + h_k}{\sum_{i=1}^{l-1} \beta h_i + \sum_{i=1}^{l-1} c_i + \sum_{i=l}^{k} e_i + h_k}\;.
$$

\begin{fact}
\label{fact:subtracting_from_up_and_down}
For any $\gamma,x,m > 0$ and $n>x$, either $\frac n m <
\frac 1 \gamma$ or $\frac{n -x }{m - \gamma x} \geq \frac n m$.
\end{fact}

We upper-bound $\calR_{1,l}$. All the following inequalities follow from
Fact \ref{fact:subtracting_from_up_and_down} (for $\gamma = 1$ or $\gamma = \beta$):
\begin{align*}
\calR_{1,l}
\leq &\;
\frac{\sum_{i=1}^{l-1} h_i  + (k-l+2) \beta h_l +
\sum_{i=l}^{k} e_i + h_k}{\sum_{i=1}^{l-1} \beta h_i + \sum_{i=l}^{k} e_i + h_k} \\
\leq &\;
\frac{\sum_{i=1}^{l-1} h_i  + (k-l+2)\beta h_l + \sum_{i=l}^{k} \beta h_i + h_k}{\sum_{i=1}^{l-1} \beta h_i + \sum_{i=l}^k \beta h_i + h_k} \\
\leq &\;
\frac{(k-l+2)\beta h_l + \sum_{i=l}^{k} \beta h_i + h_k}{\sum_{i=l}^k \beta h_i + h_k} \\
= &\; 1 + \frac{(k-l+2)\beta h_l}{\sum_{i=l}^k \beta h_i + h_k} \enspace.
\end{align*}
Fix $k$, $l$ and $h_l$. The fraction above is maximal when $h_{l+1}, \ldots
h_k$ are minimal. As $\alpha h_{i+1} \geq h_i$, minimal $h_l , \ldots h_k$
form a geometric progression with a common ratio of $\alpha^{-1}$. Thus by
taking $h_i=\alpha^{k-i} h_k$ for $i \geq l$ and fixing $m := k-l+1 \geq 1$,
we obtain
$$
\calR_{1,l} \leq 1 + \frac{(k-l+2)\beta h_l}{\sum_{i=l}^k \beta h_i + h_k} \leq
1 + \frac{(m+1)\beta \alpha^{m-1} }{\sum_{i=0}^{m-1} \beta \alpha^i + 1} = {\cal B}_m \;.
$$
For $m=1$ we have
$
{\cal B}_1 = 1+ 2 \beta/(1+\beta) = 1.8$, and
for $m=2$ we have
${\cal B}_2 = 1 + 3 \beta \alpha/(\beta ( 1 + \alpha ) + 1)
= 1 + 9/13 <1.8$.
As for $m \geq 3$, we show that ${\cal B}_m < {\cal B}_{m-1} $:
\begin{multline*}
{\cal B} _m -1 = \frac{(m+1)\beta \alpha^{m-1} }{\beta  \sum_{i=0}^{m-1} \alpha^i + 1 } <
\frac{(m+1)\beta \alpha^{m-1} }{\beta  \sum_{i=0}^{m-2} \alpha^i + 1 }= \\
=\alpha \cdot \frac{m+1}{m} \cdot \frac{m\beta \alpha^{m-2} }{\beta  \sum_{i=0}^{m-2} \alpha^i + 1 }\leq
\frac{m\beta \alpha^{m-2} }{\beta  \sum_{i=0}^{m-2} \alpha^i + 1 } = {\cal B}_{m-1} -1\; .
\end{multline*}
Thus all $\mathcal{R}_{1,l}$ are upper-bounded by $1.8$.

\noindent \mycase{2b}
Suppose $e_k \beforeeq a \before h_k$.
The adversary collected only the items preceding
$h_k$ in the queue, thus gaining at most $h_1 + \ldots + h_{k} + 2h_k$,
as in step $i \in \{ 1 \ldots k-1 \}$ it can collect item of weight greater than
$h_i$, in steps $k,k+1$ it can collect items with weight almost $h_k$ and
gain credit for $h_k$.
Hence the competitive ratio of {\FIFOQueEH} can be upper-bounded by
$$
\calR_{2} = \frac{\sum_{i=1}^k h_i + 2 h_{k}}{\sum_{i=1}^k e_i +  h_k}
\leq \frac{\sum_{i=1}^k h_i + 2 h_{k}}{\sum_{i=1}^k \beta h_i +  h_k} =
{\cal B}
$$
and by Fact \ref{fact:subtracting_from_up_and_down} for $\gamma = \beta$ either
${\cal B} < \frac 1 \beta = 1.5$ or
$$
{\cal B} \leq \frac{h_k + 2 h_{k}}{\beta h_k +  h_k} = \frac{3}{1+\beta} = 1.8
\;.
$$

\noindent \mycase{2c}
Suppose $h_k \beforeeq a$. Then he cannot get credit for
any item taken by the algorithm, due to the EEF property. In this case
the adversary's gain is at most
$h_1 + \ldots + h_{k+1}$,
which gives smaller ratio than the one obtained
in case (2b), as $\alpha > \frac{1}{2}$.
This concludes the proof.
\qed
\end{proof}

\begin{proof}[of Theorem \ref{thm: tight analysis of FIFOQueEH}]
Consider two instances. In the first one, we have
$a \before b \before c \before d$, with $w_a=w_b=\beta - \epsilon$,
$w_c=\beta$, $w_d=1$. Item $a$ is deleted right after the first step, and item
$b$ right after the second step. {\FIFOQueEH} collects $c$ in the first step and $d$
in the second step, while the adversary collects all four items.
Adversary's gain is $1+3\beta-2\epsilon$
and the {\FIFOQueEH}'s gain is $1 + \beta$. Thus the competitive ratio is arbitrarily
close to $(3 \beta + 1)/(\beta + 1)$.

In the second instance,
$a \before b \before c \before d$, with $w_a=\beta$, $w_b=w_c=1-\epsilon$,  $w_d=1$. Items $a$, $b$ and $c$ are deleted right after the second step.
{\FIFOQueEH} collects $a$ in the first step and $d$ in the second step,
so its gain is $1+ \beta$. The adversary collects $b$, $c$, $d$ (in this
order), so his gain is $3-2\epsilon$. Thus the competitive ratio is
arbitrarily close to $3/(1+\beta)$.

From these two instances,
we get that the ratio of {\FIFOQueEH} is at least $\max\braced{3\beta+1,3}/(\beta+1)$,
and this quantity is at least $1.8$ for any $\beta$.
\qed
\end{proof}

\subsection{Upper Bound for Non-Decreasing Weights}

\paragraph{Set dominance relation.}
Let $X,Y$ be two finite sets of numbers.
We say that $X$ \emph{dominates} $Y$, denoted $X\succeq Y$,
if either $Y=\emptyset$ or $\max X \geq \max Y$ and
$(X-\max X) \succeq (Y- \max Y)$. Note that we do not
require that $|X| = |Y|$. In particular, $X\succeq \emptyset$,
for any $X$.

For any set $T$ and a number $u$,
let $\sharp_u(T) = |\braced{t\in T \suchthat t\ge u}|$.
We show that the majorization can be described in terms of $\sharp_u$
The following lemma is routine and we omit the proof.
\begin{lemma}\label{fact:dom}
The following three conditions are equivalent:
\begin{description}
	\item{(i)} $X\succeq Y$,
	\item{(ii)} There is an injection $f: Y\to X$ such that
				$f(y) \ge y$ for all $y$.
	\item{(iii)} For every $x$ we have $\sharp_x(X) \geq \sharp_x(Y)$.
\end{description}
\end{lemma}

\begin{lemma}\label{lem: dominance update}
Suppose that $X\succeq Y \neq \emptyset$. Then
\begin{description}
	\item{(i)} $X-\min X \succeq Y-\min Y$.
	\item{(ii)} If $x\in X\cap Y$ then $X-x\succeq Y-x$.
	\item{(iii)} If $X,Y\subseteq Z$, $y\in Z-Y$,
	and $x \ge \max\braced{z\in Z-X \suchthat z \le y}$
				then $X\cup x \succeq Y\cup y$.
				(In particular, this holds for $x\ge y$.)
\end{description}
\end{lemma}

\begin{proof}
Parts (i) and (ii) are straightforward, so we only show (iii).
Let $x' = \max\braced{z\in Z-X \suchthat z \le y}$. It is
sufficient to show (iii) for $x = x'$.

We use Lemma~\ref{fact:dom}(iii).
For $u \le x'$, since $X\succeq Y$, we have
$\sharp_u(X\cup x') = \sharp_u(X)+1
			\ge \sharp_u(Y)+1
			= \sharp_u(Y\cup y)$.
For $u > y$, we have
$\sharp_u(X\cup x') = \sharp_u(X)
				\ge \sharp_u(Y)
			= \sharp_u(Y\cup y)$.
Suppose $x' < u \le y$.
Since $y\in Z-Y$ and $X \cap [u,y] = Z\cap [u,y]$,
we have $|X\cap [u,y]| > |Y\cap [u,y]|$, and therefore
$\sharp_u(X\cup x') = \sharp_u(X)
				\ge \sharp_u(Y) + 1
				 = \sharp_u(Y\cup y)$.
\qed
\end{proof}


For simplicity, we assume that all weights are different.
If there are equal weights, we can perturb them slightly or
extend the weight ordering using the item indices.
For sets $B$ and $C$ of items, we say that $B$
dominates $C$, writing $B\succeq C$, if $\braced{w_b\suchthat b\in B}$
dominates $\braced{w_c\suchthat c\in C}$.
We write $B\succeq aC$ if
$\braced{w_b\suchthat b\in B}$
dominates $\braced{aw_c\suchthat c\in C}$.

Assume that the adversary has EEF property.
Consider the behavior of the adversary. By the EEF property,
for any $i,j\in \calS$,
we can assume that if $j$ is the item collected in
step $t$ and $i\before j$ then the adversary will
not collect item $i$ in the future. Thus, instead of considering
the whole set of pending adversary items, we can
restrict ourselves only to those that are after the one he
collected last. Let $C_t$ be the set of these items.
We update $C_t$ as follows: at each step, the deleted items
are removed from $C_t$ and released items are added to $C_t$.
Then, if the adversary collects an item $j$, we remove
all items $i\beforeeq j$ from $C_t$.

To organize the accounting so that we can pay for these
extra items, we do two things. One, we do not immediately
give {\MarkAndPick} credit for the collected items. We give
it credit for collecting $i$ only at the time when the adversary
collects an item at least as heavy as $i$.
This is formalized below in Lemma~\ref{lem: det queue incr weight invariant}
where the collected items $i$ that are heavier than the
maximum adversary item contribute to the both sides of
invariant (a) (they get included in $L_t$ and $L'_t$),
and only when the adversary collects an item greater
than $i$, item $i$ will be removed from $L'_t$ and
contribute to preserving the invariant.

Next, we keep track of the adversary's extra items so that
we can pay for them later. When the adversary collects
an extra item $j$ and we mark $h$, we know that
$w_j\le w_h/\phi$. We store $j$ in a separate set $E'$
and $h$ in $M'$. Later, when there is an idle step
at which the adversary collects $h$, since $h$ has
already been collected by the algorithm, its
budget $\phi w_h$ pays both for $w_h$ (for the mark on $h$)
and for $w_j$. Once the adversary "consumed" the extra
step by collecting $h$, we move $j$ to the set $E$ of the
extra adversary items that we already paid for.

We represent our invariants in terms of the
domination relations between some varying sets of items.
The reason for this is that items can be added and
removed from these sets in this process, and it does not
seem possible to maintain appropriate bounds only
between the total weights of these sets.

\begin{figure}[h]
\begin{center}
\includegraphics[width=3.25in]{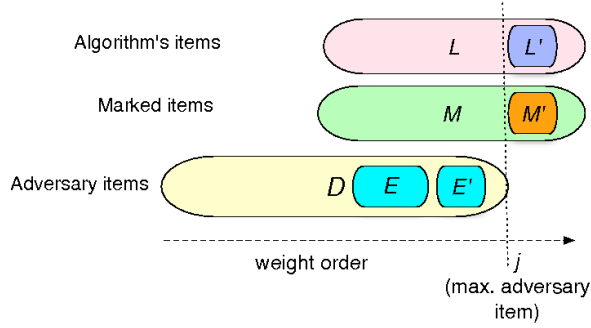} %
\caption{Notation.}\label{fig: notation}
\end{center}
\end{figure}


\smallskip
\emph{Notation:} Symbols
$D_t$, $M_t$, and $L_t$ represent respectively the
sets of items collected by the adversary, marked by the
algorithm, and collected by Algorithm~{\MarkAndPick} up to and including
step $t$. $L'_t = L_t\cap C_t$ is the set of algorithm's items
that the adversary may still collect in the future.
Let also $e_t = |D_t| - |M_t|$ and $\ell_t = |L'_t|$.
Figure~\ref{fig: notation} illustrates this notation, as
well as the sets introduced in the lemma below.

\begin{lemma}\label{lem: det queue incr weight invariant}
For each time step $t$, there exist
disjoint sets $E_t, E'_t\subseteq D_t$ with $|E_t| = e_t$ and
$|E'_t|=\ell_t$,
and a set $M'_t\subseteq M_t$ with $|M'_t| = \ell_t$, such that
\begin{description}
	\item{(a)} $\phi w(L_t- L'_t) \ge w(M_t -M'_t) + w(E_t)$,
	\item{(b)} $\phi L'_t \succeq M'_t \succeq L'_t$,
	\item{(c)} $M_t \succeq \left(D_t - E_t -E'_t\right) \cup L'_t$, and
	\item{(d)} $M'_t \succeq \phi E'_t$.
\end{description}
\end{lemma}

\begin{proof}
We show that the invariant in the lemma is preserved.
For simplicity, we omit the subscript $t$ and
write $D= D_t$, $M= M_t$, etc. Also,
let $\Delta w(D) = w(D_{t+1}) - w(D_t)$, $\Delta w(M) = w(M_{t+1}) - w(M_t)$,
and so on.

We view the process as follows: at each step,
\begin{description}
	
	\item{(I)} The adversary first inserts items into $\calS$,
	\item{(II)} then he selects the item $j$ to be collected,
	\item{(III)} next, the adversary deletes some items from $\calS$ (of course,
			only the items that are before $j$ in $\calS$ can be deleted),
	\item{(IV)} finally, both the adversary and the algorithm collect their items.

\end{description}

In order to show (a), we need to show that
\begin{eqnarray}
	\phi \Delta w(L) +\Delta w(M')
			&\ge& \phi\Delta w(L') + \Delta w(M) + \Delta w(E).
		\label{eqn: amortization}
\end{eqnarray}
In addition, we need to show that (b), (c) and (d) are preserved.

We look at all sub-steps separately.
(I) Insertions do not affect the invariants. (None of the
sets $M$, $D$, $L$, $L'$ changes, and we do not change
sets $M'$, $E$, and $E'$.)
In (II), suppose the adversary selects $j$ and $j'$ was
an item collected by the adversary in the previous step.
There may have been some items $i$, $j'\before i\before j$, that
were in $L'$. Since these items are now removed from $C$,
they are also removed from $L'$, and we need to update
$M'$ and $E'$ so that they have the same cardinality as $L'$,
and in such a way that the invariants are preserved.
Let $i\in L'$ be such an item with
minimum weight, $g$ the minimum-weight item in $M'$ and $e$
the minimum-weight item in $E'$.
We remove $i$ from $L'$, $g$ from $M'$ and $e$ from $E'$.
Since $\phi w_i\ge w_g$, by (b), we have
$\phi \Delta w(L) + \Delta w(M') = 0 -w_g
			\ge -\phi w_i + 0 + 0
			 = \phi\Delta w(L') + \Delta w(M) + \Delta w(E)$,
so (a) is preserved. Invariants (b) and (d) are preserved
because we remove the minimum items from $L'$, $M'$, and
$E'$. In (c), the left-hand side does not change and on
the right-hand side we remove $i$ from $L'$ and add $e$ to
$D-E-E'$, and by (b) and (d) we have $w_i\ge w_g/\phi \ge w_e$,
so the right-hand side cannot increase.
In sub-step (III), deletions do not affect the invariants.

The rest of the proof is devoted to sub-step (IV).
We examine (\ref{eqn: amortization}) and the changes
in (b), (c) and (d)
due to the algorithm and the adversary collecting their items.

\smallskip
\noindent
\mycase{A}
There is at least one pending item.
The algorithm marks $h$ and collects the earliest pending
item $i$ such that $w_i\ge w_{h}/\phi$. Thus $h$ is added to $M$
and $i$ is added to $L$. We do not change $E$.
We have some sub-cases.

\smallskip
\noindent
\mycase{A.1}
$i \beforeeq j$. Then $i$ is not added to $L'$, and we do not change $E'$.
Since $\phi\Delta w(L) = \phi w_i \ge w_h = \Delta w(M)$
and $\Delta w(E) = 0$, to prove (\ref{eqn: amortization})
it is now sufficient to show that
\begin{eqnarray}
\Delta w(M') &\ge& \phi\Delta w(L').
		\label{eqn: amortization A1}
\end{eqnarray}
If $j\notin L$ (note that this included the case $i=j$),
then $L'$ does not change and we do not change $M'$, so
(\ref{eqn: amortization A1}) is trivial. In (b) and (d) nothing
changes. We add the maximum unmarked item $h$ to $M$ and $j$ to $D$,
so (c) follows from Lemma~\ref{lem: dominance update}(iii).

If $i\neq j$ and $j\in L$ and
then let $g$ be the minimum item in $M'$ and
$e$ the minimum item in $E'$.
Item $j$ is removed from $L'$ and we remove $g$ from $M'$ and
$e$ from $E'$.
Since, by (b), $\phi w_j\ge w_g$, we have
$\Delta w(M') = -w_g \ge -\phi w_j = \phi\Delta w(L')$, so
(\ref{eqn: amortization A1}) holds.
Since $j$, $g$ and $e$ are minimal, (b) and (d) are preserved.
In (c), moving $j$ from $L'$ to $D$ does not change the right-hand
side. We also add $h$ to $M$ and
$e$ to $D-E-E'\cup L'$, so (c) is preserved because
of the choice of $h$ and Lemma~\ref{lem: dominance update}(iii).

\smallskip
\noindent
\mycase{A.2}
$i \after j$ and $j\notin L$.
Then $i$ is added to $L'$. We also add $j$ to $E'$.
Since $\Delta w(L) = \Delta w(L') = w_i$ and $\Delta w(M) = w_h$,
to show (\ref{eqn: amortization}) it is sufficient to show that
\begin{eqnarray}
\Delta w(M') &\ge& w_h.
		\label{eqn: amortization A2}
\end{eqnarray}
Since $|L'|$ increased,
we need to add one item $f$ to $M'$. We choose this $f$ as
follows: If $i\beforeeq h$, we choose $f=h$.
Otherwise, if $i \after h$ then, by the choice of $h$, we get that
all items $h\beforeeq f\beforeeq i$ are now marked.
In this case,
we choose the largest $f\beforeeq i$ such that $f\notin M'$
and add it to $M'$. Since $h\in M-M'$, $h$ itself is a
candidate for $f$, so we have $h \beforeeq f \beforeeq i$.

Note that, in this case, by the choice of
$i$ (as the earliest pending item with weight at least $w_h/\phi$),
$j\before i$ and $j\notin L$, we have $w_f \ge w_h \ge \phi w_j$.
In particular, this means that $j \before h$ and that
all items $h\beforeeq f'\before i$ are in $L'$.
	
Since $w_f\ge w_h$, (\ref{eqn: amortization A2}) is trivial.
Invariant (d) is also quite easy, since $w_f\ge \phi w_j$,
by the previous paragraph.
In (c), adding $j$ to $D$ and $E'$ does not change the right-hand side.
We also add $h$ to $M$ and  $i$ to $L'$, which preserves (c)
by the choice of $h$ and Lemma~\ref{lem: dominance update}(iii).

To show (b), if $i\beforeeq h$ then $f=h$ and,
since $\phi w_i \ge w_h \ge w_i$, invariant (b) is preserved.
If $i\after h$, then,
$\phi w_i \ge w_i \ge w_f$, so
the first part of (b) is preserved.
That the second part of (b) is preserved follows
from the choice of $f$ and
Lemma~\ref{lem: dominance update}(iii).

\smallskip
\noindent
\mycase{A.3}
$i\after j$ and $j\in L$.
Then we remove $j$ from $L'$ and add $i$.
We thus have $\Delta w(L) = w_i$,
$\Delta w(L') = w_i - w_j$ and $\Delta w(M) = w_h$.
Thus to show (\ref{eqn: amortization}) it is sufficient to show that
\begin{eqnarray}
\Delta w(M')+ \phi w_j &\ge& w_h.
		\label{eqn: amortization A3}
\end{eqnarray}
We do not change $E'$.
To update $M'$, we proceed as follows.
Let $g$ be the lightest item in $M'$. Since $j$ is
the minimal element of $L'$,
(b) implies $w_j\le w_g\le \phi w_j$.
We first remove $g$ from $M'$.
Next, we proceed similarly as in the previous case,
looking for an item $f$ that we can add to
$M'$ to compensate for removing $g$ (since $|M'|$ cannot
change in this case.)
Let $h' = \max(g,h)$. Note that $h'\in M-M'$.
If $i\beforeeq h'$, we choose $f=h'$.
Otherwise, if $i \after h'$ then, by the choice of $h'$, we get that
all active items $h'\beforeeq f\beforeeq i$ are marked. In this case,
we choose the largest $f\beforeeq i$ such that $f\notin M'$
and add it to $M'$. Since $h'\in M- M'$, $h'$ itself is a
candidate for $f$, so we have $h' \beforeeq f \beforeeq i$.

Note that, in this case, by $j\beforeeq g$, all items
$h'\beforeeq f\before i$ are active and,
by the choice of $i$, they are all in $L'$.
	
Now, in (\ref{eqn: amortization A3}) we have
$\Delta w(M')  + \phi w_j
	= (-w_g+ w_f) + \phi w_j
	\ge w_f \ge w_{h'} \ge w_h$.
In (d), the left-hand side can only increase
(since $f\ge g$) and the right-hand side does not change.
In (c), moving $j$ from $L'$ to $D$ does not change
the right-hand side. We also
added $h$ to the left-hand side and
$i$ to $L'$ on the right-hand side, so (c) is
preserved by the choice of $h$ and
Lemma~\ref{lem: dominance update}(iii).

In (b), removing $j$ from $L'$ and $g$ from $M'$ does
not affect the invariant. Then we add $f$ to $M'$ and
$i$ to $L'$. By the algorithm, we have $\phi w_i \ge w_h$,
while by the case assumption and (b), we have
$\phi w_i \ge \phi w_j \ge w_g$. Therefore
$\phi w_i \ge w_{h'}$. Since either $f=h'$
or $f\beforeeq i$, this implies
$\phi w_i \ge w_f$, showing that the first
inequality in (b) is preserved.
The second part of (b) follows again from
the choice of $f$ and Lemma~\ref{lem: dominance update}(iii).

\smallskip
\noindent
\mycase{B}
There are no pending items for the algorithm.
It means that $L'$ contains all active
items $i\aftereq j$, including $j$.
By the weight ordering
assumption and the second part of (b)
this implies that $L' = M'$.
Since the adversary collects an item and the
algorithm does not, $e=|D|-|M|$ increases by $1$, so
we also need to add an item to $E$. Let $b$ be the
minimum-weight item in $E'$.
We do this: we remove $j$ from $M'$ and from $L'$ and
we move $b$ from $E'$ to $E$.
Using the choice of $b$ and (d) we have $w_j\ge \phi w_b$,
so
\begin{multline*}
\phi\Delta w(L) + \Delta w(M') \;=\; 0 + (-w_j) \;=\; -\phi w_j + w_j/\phi \;\geq \\
				 \geq\; -\phi w_j + 0 + w_b \;=\; \phi \Delta w(L') + \Delta w(M) + \Delta w(E),
\end{multline*}
and thus (\ref{eqn: amortization}) holds.
By the choice of $j$ and $e$ as the minimum items in $L'$ and $E'$,
respectively, invariants (b) and (d) are preserved. In (c),
$j$ moves from $L'$ to $D$, and
$b$ moves from $E'$ to $E$, so the right-hand side does
not change.
\qed
\end{proof}


\begin{proof}[of Theorem \ref{thm: det queue incr weight}]
By the invariants of Lemma~\ref{lem: det queue incr weight invariant},
at each time step it holds that
\begin{eqnarray*}
\phi w(L_t) &\ge& [\phi w(L'_t)-w(M'_t)] + w(M_t) + w(E_t)
				\\
 		&\ge& 0 + [w(D_t - E_t - E'_t) + w(L'_t)] + w(E_t)
				\\
		&=&  w(D_t) + w(L'_t) - w(E'_t)
				\\
		&\ge&  w(D_t) + w(M'_t)/\phi - w(E'_t)
				\\
		&\ge& w(D_t) \enspace,
\end{eqnarray*}
and the $\phi$-competitiveness follows. \qed
\end{proof}

\section{Randomized Algorithms for Dynamic Queues}
\begin{proof}[of Theorem \ref{thm: dynamic queue rand m-less lower bound}]
Fix some online memoryless randomized algorithm $\calA$. Recall that by a memoryless
algorithm we mean an algorithm that makes a decision on which item to
collect based only on the weights of the pending items.

We consider the following scheme. Let $a > 1$ be a constant which we specify later
and $n$ be a fixed integer. At the beginning, the adversary inserts items
$a^0,a^1,\ldots,a^n$. (To simplify notation, in this proof we identify
items with their weights.) In our construction we assure that in each step,
the list of items which are pending for $\calA$ is equal to $(a^0,a^1,\ldots,a^n)$.
Since $\calA$ is memoryless, in each step it uses the same probability distribution
$(q_j)_{j=0}^n$, where $q_j$ is the probability of collecting item $a^j$.
As the algorithm always makes a move, $\sum_{i=0}^n q_i = 1$.

We consider $n+1$ strategies for an adversary, numbered $0,1,\ldots,n$.
The $k$-th strategy is as follows: in each step collect $a^k$,
delete all items $a^0,a^1,...,a^k$, and then issue new copies of all these items.
Additionally, if the algorithm collected $a^j$ for some
$j>k$, then the adversary issues a new copy of $a^j$ as well.
This way, in each step exactly one copy of each $a^j$ is
pending for the algorithm, while the adversary accumulates
copies of the items $a^j$ for $j>k$.

This step is repeated $T \gg n$ times, and after the last step
the adversary collects all uncollected items.
Since $T \gg n$, we only need to focus on the
expected amortized profits in a single step.

We look at the gains of $\calA$ and the adversary in a single step.
If the adversary chooses strategy $k$, then it gains $a^k$. Additionally,
at the end it collects the item collected by the algorithm if this item is
greater than $a^k$. Thus, its amortized expected
gain in single step is $a^k + \sum_{i>k} q_i a^i$. The expected
gain of $\calA$ is $\sum_i q_i a^i$.

For any probability distribution $(q_j)_{j=0}^n$ of the algorithm, the adversary chooses
a strategy $k$ which maximizes the competitive ratio. Thus, the competitive
ratio of $\calA$ is is at least
\[		
	\calR = \max_k \braced{ \frac{a^k + \sum_{j > k} q_j a^j}{\sum_j q_j a^j}}
		 \;\geq\; \sum_k v_k  \frac{a^k + \sum_{j > k} q_j a^j}{\sum_j q_j a^j} \enspace,
\]
for any coefficients $v_0,...,v_k\ge 0$ such that $\sum_k v_k = 1$.
Note that the latter term corresponds to the ratio forced
by a randomized adversary who chooses $k$ with probability $v_k$. In particular, we may
choose $v_k$ to be
the value for which the competitive ratio of such a randomized adversary strategy against
\emph{any deterministic} algorithm is the same. After solving the set of equations we get
\[
	v_k = \begin{cases}
		\frac{1}{M}  a^{n-k}  ( a-1 ) 	& \textnormal{if $k < n$} \\
		\frac{1}{M}  \left( a - n  (a-1) \right) & \textnormal{if $k = n$} \\
 	\end{cases}
	\quad\quad\text{ where }\; M = a^{n+1} - n  (a-1)
	\enspace.
\]
For these values of $v_k$ we get
\begin{align*}
\textstyle M \calR \sum_j q_j a^j
	\geq & \; \sum_k M v_k a^k + \sum_k M v_k \sum_{j>k} q_j a^j
	\\
\textstyle
	= & \; \sum_{k=0}^{n-1} M v_k a^k + M v_n a^n + \sum_j q_j a^j \sum_{k < j} M  v_k
	= \; a^{n+1}  \sum_j q_j  a^j \enspace.
\end{align*}
Therefore, $\calR \geq \frac{a^{n+1}}{M}$. This bound is maximized for $a = 1+1/n$, for which 
\[
\calR \;\geq\;
\left(1+\frac{1}{n}\right)^{n+1}\left(\left(1+\frac{1}{n}\right)^{n+1}
- 1\right)^{-1}
\]
which tends to $\frac{\e}{\e-1}$
 as $n \to \infty$.
\qed
\end{proof}

\end{appendix}

\end{document}